\documentclass[journal,12pt,onecolumn,draftcls]{IEEEtran}
\usepackage{cite}
\ifCLASSINFOpdf
  \usepackage[pdftex]{graphicx}
  \graphicspath{{../pdf/}{../jpeg/}}
  \DeclareGraphicsExtensions{.pdf,.jpeg,.png}
\else
  \usepackage[dvips]{graphicx}
    \usepackage{color}
  \usepackage{epsfig}
  \graphicspath{{../eps/}}
   \DeclareGraphicsExtensions{.eps}
\fi
\usepackage[cmex10]{amsmath}
\usepackage{amssymb}

\DeclareMathOperator{\diag}{diag}
\DeclareMathOperator{\tr}{tr}
\DeclareMathOperator{\expec}{\mathbb{E}}
\DeclareMathOperator{\var}{var}

\begin{document}
\title{Pilot Contamination and Precoding in Multi-Cell TDD Systems}

\author{\IEEEauthorblockN{Jubin Jose,~\IEEEmembership{Student Member,~IEEE,}
Alexei Ashikhmin,~\IEEEmembership{Senior Member,~IEEE,}\\
Thomas L. Marzetta,~\IEEEmembership{Fellow,~IEEE,}
and Sriram Vishwanath~\IEEEmembership{Senior Member,~IEEE}}
\thanks{Results in this paper were presented in part at the IEEE International Symposium on Information Theory (ISIT) \cite{Jose_AMV_09}. J. Jose and S. Vishwanath are with the Department of Electrical and Computer Engineering, The University of Texas at Austin, Austin, TX 78712 USA (email: jubin@austin.utexas.edu; sriram@austin.utexas.edu). A. Ashikhmin and T. L. Marzetta are with Bell Laboratories, Alcatel-Lucent Inc., Murray Hill, NJ 07974 USA (email: aea@research.bell-labs.com; tlm@research.bell-labs.com).}% <-this % stops a space
}
%\thanks{\IEEEauthorrefmark{1} The authors were supported in part by the DoD and the ARO Young Investigator Program (YIP).}}
\maketitle

\newtheorem{theorem}{Theorem}
\newtheorem{corollary}[theorem]{Corollary}
\newtheorem{lemma}[theorem]{Lemma}
\newtheorem{remark}{Remark}

\begin{abstract}
%\boldmath
This paper considers a multi-cell multiple antenna system with precoding used at the base stations for downlink transmission. For precoding at the base stations, channel state information (CSI) is essential at the base stations. A popular technique for obtaining this CSI in time division duplex (TDD) systems is uplink training by utilizing the reciprocity of the wireless medium. This paper mathematically characterizes the impact that uplink training has on the performance of such multi-cell multiple antenna systems. When non-orthogonal training sequences are used for uplink training, the paper shows that the precoding matrix used by the base station in one cell becomes corrupted by the channel between that base station and the users in other cells in an undesirable manner. This paper analyzes this fundamental problem of pilot contamination in multi-cell systems. Furthermore, it develops a new multi-cell MMSE-based precoding method that mitigate this problem. In addition to being a linear precoding method, this precoding method has a simple closed-form expression that results from an intuitive optimization problem formulation. Numerical results show significant performance gains compared to certain popular single-cell precoding methods.
\end{abstract}
\IEEEpeerreviewmaketitle
\section{Introduction}
\label{intro}
Multiple antennas, especially at the base-station, have now become an accepted (and in someways, a central) feature of cellular networks. These networks have been studied extensively over the past one and a half decades (see\cite{MIMO_wireless_text} and references therein). It is now well understood that channel state information (CSI) at the base station is an essential component when trying to maximize network throughput. Systems with varying degrees of  CSI have been studied in great detail in literature. The primary framework under which these have been studied is frequency division duplex (FDD) systems, where the CSI is typically obtained through  (limited) feedback. There is a rich body of work in jointly designing this feedback mechanism with (pre)coding strategies to maximize throughput in MIMO downlink \cite{SH05,DLZ05,N06,YJG07,Ashikhmin_ISIT07,HHA07:SDMA_feedback_rate}. Time division duplex (TDD) systems, however, have a fundamentally different architecture from the ones studied in FDD systems \cite{Marzetta:howmuchtraining,JAWV08}. The goal of this paper is to develop a clear understanding of mechanisms for acquiring CSI and subsequently designing precoding strategies for multi-cell MIMO TDD systems.

An important distinguishing feature of TDD systems is the notion of {\em reciprocity}, where  the reverse channel is used as an estimate of the forward channel. Arguably, this is one of the best advantages of a TDD architecture, as it eliminates the need for feedback, and uplink training together with the reciprocity of the wireless medium \cite{MH06} is sufficient to provide us with the desired CSI. However, as we see next, this channel estimate is not without issues that must be addressed before it proves useful.

In this paper, we consider uplink training and transmit precoding in a multi-cell scenario with $L$ cells, where each cell consists of a base station with $M$ antennas and $K$ users with single antenna each. The impact of uplink training on the resulting channel estimate (and thus system performance) in the multi-cell scenario is significantly different from that in a single-cell scenario. In the multi-cell scenario, non-orthogonal training sequences (pilots) must be utilized, as orthogonal pilots would need to be least $K \times L$ symbols long which is infeasible for large $L$. In particular, short channel coherence times due to mobility do not allow for such long training sequences.

This non-orthogonal nature causes {\em pilot contamination}, which is encountered only when analyzing a multi-cell MIMO system with training, and is lost when narrowing focus to a single-cell setting or to a multi-cell setting where channel information is assumed available at no cost. Pilot contamination occurs when the channel estimate at the base station in one cell becomes polluted by users from other cells. Thus, our goal in this paper is, first, to study the impact of pilot contamination (and thus achievable rates), and then, to develop methods that mitigate this contamination. We note that pilot contamination must also figure in Cooperative MIMO (also called Network MIMO \cite{Foschini:coordinatingnetwork,VLV07}) where clusters of base stations are wired together to create distributed arrays, and where pilots must be re-used over multiple clusters.

%figure1
\begin{figure}[ht]
\centering
\scalebox{0.6}{\input{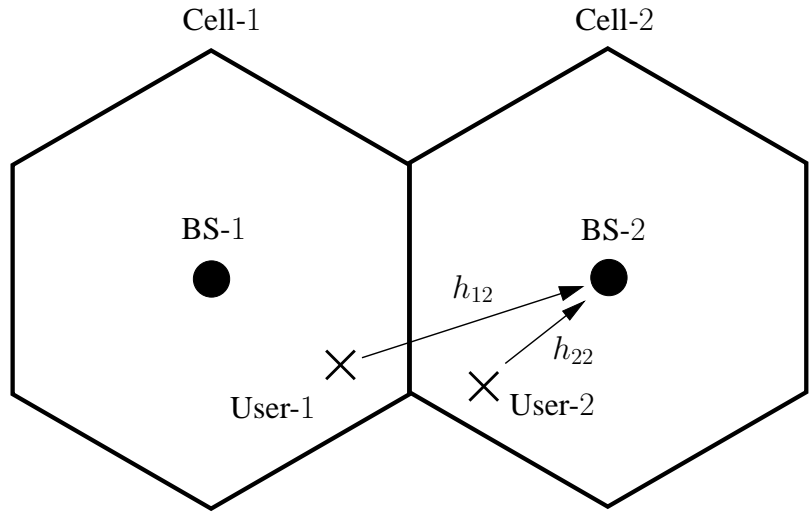}}
\DeclareGraphicsExtensions. \caption{A two-cell example with one user in each cell. Both users transmit non-orthogonal pilots during uplink training, which leads to pilot contamination at both the base stations.} \label{two_cell}
\end{figure}

The fundamental problem associated with pilot contamination is evident even in the simple multi-cell scenario shown in Figure \ref{two_cell}. Consider two cells $i\in\{1,2\}$, each consisting of one base station and one user. Let $\underline{h}_{ij}$ denote the channel between the base station in the $i$-th cell and the user in the $j$-th cell. Let the training sequences used by both the users be same. In this case, the MMSE channel estimate of $\underline{h}_{22}$ at the base station in the $2$-nd cell is $\hat{\underline{h}}_{22}=c_1 \underline{h}_{12} + c_2 \underline{h}_{22} + c \underline{w}$. Here $c_1$, $c_2$ and $c$ are constants that depend on the propagation factors and the transmit powers of mobiles,  and $\underline{w}$ is $\mathcal{CN}(0,\mathbf{I})$ additive noise. The base station in the $2$-nd cell uses this channel estimate to form a precoding vector $\underline{a}_2=f(\hat{\underline{h}}_{22}),$ which is usually aligned with the channel estimate, that is $\underline{a}_2=const\cdot \hat{\underline{h}}_{22}$. However, by doing this, the base station (partially) aligns the transmitted signal with both $\underline{h}_{22}$ (which is desirable) and $\underline{h}_{12}$ (which is undesirable). Both signal ($\underline{h}_{22}\underline{a}_2^{\dagger}$) and interference ($\underline{h}_{12}\underline{a}_2^{\dagger}$) statistically behave similarly. Therefore, the general assumption that the precoding vector used by a base station in one cell is uncorrelated with the channel to users in other cells is not valid with uplink training using non-orthogonal training sequences. This  fundamental problem is studied in further detail in the rest of this paper.

To perform this analysis, we first develop analytical expressions using techniques similar to those used in \cite{Marzetta:howmuchtraining,JAWV08}. For the setting with one user in every cell, we derive closed-form expressions for achievable rates. These closed-form expressions allow us to determine the extent to which pilot contamination impacts system performance. In particular, we show that the achievable rates can saturate with the number of antennas at the base station $M$. This analysis will allow system designers to determine the appropriate frequency/time/pilot reuse factor to maximize system throughput in the presence of pilot contamination.

In the multi-cell scenario, there has been significant work on utilizing coordination among base stations \cite{SSZ04,ZD04,Foschini:coordinatingnetwork,VLV07} when CSI is available. This existing body of work focuses on the gain that can be obtained through coordination of the base stations. Dirty paper coding  based approaches and joint beamforming/precoding approaches are considered in \cite{ZD04}. Linear precoding methods for clustered networks with full intra-cluster coordination and limited inter-cluster coordination are proposed in \cite{Zhang_CAGH_2009}. These approaches generally require ``good'' channel estimates at the base stations. Due to non-orthogonal training sequences, the resulting channel estimate (of the channel between a base station and all users) can be shown to be rank deficient. We develop a multi-cell MMSE-based precoding method that depends on the set of training sequences assigned to the users. Note that this MMSE-based precoding is for the general setting with multiple users in every cell. Our approach does not need coordination between base stations required by the joint precoding techniques. When coordination is present, this approach can be applied at the inter-cluster level. The MMSE-based precoding derived in this paper has several advantages. In addition to being a linear precoding method, it has a simple closed-form expression that results from an intuitive optimization problem formulation. For many training sequence allocations, numerical results show that our approach gives significant gains over certain popular single-cell precoding methods including zero-forcing precoding.

\subsection{Related Work}
\label{sec:relwork}

Over the past decade, a variety of aspects of downlink and uplink transmission problems in a single cell setting have been studied. In information theoretic literature, these problems are studied as the broadcast channel (BC) and the multiple access channel (MAC) respectively. For Gaussian BC and general MAC, the problems have been studied for both single and multiple antenna cases. The sum capacity of the multi-antenna Gaussian BC has been shown to be achieved by dirty paper coding (DPC) in \cite{Caire:gaussianbc,Viswanath:sumcapacity,Vishwanath:duality,YC04}. It was shown  in \cite{Weingarten:capacity} that DPC characterizes the full capacity region of the multi-antenna Gaussian BC. These results assume perfect CSI at the base station and the users. In addition, the DPC technique is computationally challenging to implement in practice. There has been significant research focus on reducing the computational complexity at the base station and the users. In this regard, different precoding schemes with low complexity have been proposed. This body of work  \cite{Boccardi:precoding,Airy:transmit_precoding,HPS05_part2,SVH06,Shen:low_complexity} demonstratesƒs that sum rates close to sum capacity can be achieved with much lower computational complexity. However, these results assume perfect CSI at the base station and the users.

The problem of lack of channel CSI is usually studied by considering one of the following two settings. As discussed before, in the first setting, CSI at users is assumed to be available and a limited feedback link is assumed to exist from the users to the base station. In \cite{SH05,N06,YJG07,HHA07:SDMA_feedback_rate,Ashikhmin_ISIT07,JCCU08} such a setting is considered. In \cite{N06}, the authors show that at high signal to noise ratios (SNRs), the feedback rate required per user must grow linearly with the SNR (in dB) in order to obtain the full MIMO BC multiplexing gain. The main result in \cite{YJG07} is that the extent of CSI feedback can be reduced by exploiting multi-user diversity. In \cite{Ashikhmin_ISIT07} it is shown that nonrandom vector quantizers can significantly increase the MIMO downlink throughput.
 In \cite{HHA07:SDMA_feedback_rate}, the authors design a joint CSI quantization, beamforming and scheduling algorithm to attain optimal throughput scaling. In the next setting, time-division duplex systems are considered and channel training and estimation error are accounted for in the net achievable rate. This approach is used in \cite{Marzetta:howmuchtraining,JAWV08,GPS08}. In \cite{Marzetta:howmuchtraining}, the authors give a lower bound on sum capacity and demonstrate that it is always beneficial to increase the number of antennas at the base station. In \cite{JAWV08}, the authors study a heterogeneous user setting and present scheduling and precoding methods for this setting. In \cite{GPS08}, the authors consider two-way training and propose two variants of linear MMSE precoders as alternatives to linear zero-forcing precoder used in \cite{Marzetta:howmuchtraining}.

Given this extensive body of literature in single-cell systems, the main contribution of this paper is in understanding multi-cell systems with channel training. Its emphasis is on TDD systems, which are arguably poorly studied compared to FDD systems. Specifically, the main contributions are to demonstrate the pilot contamination problem associated with uplink training, understand its impact on the operation of multi-cell MIMO TDD cellular systems, and develop a new precoding method to mitigate this problem.

\subsection{Organization}
The rest of this paper is organized as follows. In Section \ref{sec:sysmodel}, we describe the multi-cell system model. In Section \ref{sec:commscheme}, we explain the communication scheme and the technique to obtain achievable rates. We analyze the effect of pilot contamination in Section \ref{sec:pilot_cont}, and give the details of the new precoding method in Section \ref{sec:mmseprecoding}. We present few numerical results in Section \ref{sec:numresults}. Finally, we provide our concluding remarks in Section \ref{sec:conc}.

\subsection{Notation}
We use bold font variables to denote matrices, and underline variables to denote vectors (can be row or column vectors). $(\cdot)^T$ denotes the transpose and $(\cdot)^{\dagger}$ denotes the Hermitian transpose, $\tr\{\cdot \}$ denotes the trace operation, $(\cdot)^{-1}$ denotes the inverse operation, and $\|\cdot\|$ denotes the two-norm. $\diag \{\underline{a}\}$ denotes a diagonal matrix with diagonal entries equal to the components of $\underline{a}$. $\mathbb{E}[\cdot]$ and $\var\{\cdot\}$ stand for expectation and variance operations, respectively. 

\section{Multi-Cell TDD System Model}
\label{sec:sysmodel}

%figure2
\begin{figure}[!t]
\centering
\scalebox{0.6}{\input{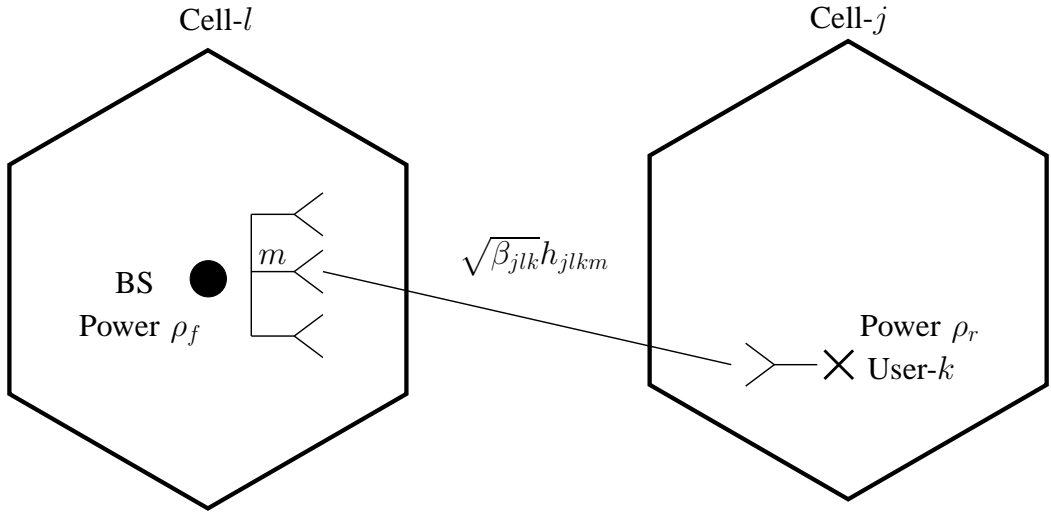}}
\caption{System model showing the base station in $l$-th cell and the $k$-th user in $j$-th cell} \label{fig:sys_model}
\end{figure}

We consider a cellular system with $L$ cells numbered $1,2,\cdots,L$. Each cell comprises of one base station with $M$ antennas and $K(\le M)$ single-antenna users. Let the average power (during transmission) at the base station be $p_f$ and the average power (during transmission) at each user be $p_r$. The propagation factor between the $m$-th base station antenna of the $l$-th cell and the $k$-th user of the $j$-th cell is $\sqrt{\beta_{jlk}} h_{jlkm}$\footnote{For compact notation, we do not separate the subscript or superscript indices using commas throughout the paper.}, where $\{\beta_{jlk}\}$ are non-negative constants and assumed to be known to everybody, and $\{h_{jlkm}\}$ are independent and identically distributed (i.i.d.) zero-mean, circularly-symmetric complex Gaussian $\mathcal{CN}(0,1)$ random variables and known to nobody. This system model is shown in Figure \ref{fig:sys_model}. The above assumptions are fairly accurate and justified due to the following reason. The $\{\beta_{jlk}\}$ values model path-loss and shadowing that change slowly and can be learned over long period of time, while the $\{h_{jlkm}\}$ values model fading that change relatively fast and must be learned and used very quickly. Since the cell layout and shadowing are captured using the constant $\{\beta_{jlk}\}$ values, for the purpose of this paper, the specific details of the cell layout and shadowing model are irrelevant. In other words, any cell layout and any shadowing model can be incorporated with the above abstraction.

We assume channel reciprocity for the forward and reverse links, i.e., the propagation factor $\sqrt{\beta_{jlk}} h_{jlkm}$ is same for both forward and reverse links, and block fading, i.e., $\{h_{jlkm}\}$ remains constant for a duration of $T$ symbols. Note that we allow for a constant factor variation in forward and reverse propagation factors through the different average power constraints at the base stations and the users. The additive noises at all terminals are i.i.d. $\mathcal{CN}(0,1)$ random variables. The system equations describing the signals received at the base station and the users are given in the next section.

\begin{remark}
The above channel model does not directly incorporate frequency-selective fading. However, we assume orthogonal frequency-division multiplexing (OFDM) operation. Therefore, in each OFDM sub-band, we can consider the above channel model. The coherence time $T$ is for the OFDM sub-band. The details of OFDM (including cyclic prefix) are omitted as this is not the main focus of this paper. The block fading model is widely used to capture channel coherence, and is known to model this fairly well. Furthermore, we assume that there is time synchronization present in the system for coherent uplink transmission.
\end{remark}

\section{Communication Scheme}
\label{sec:commscheme}

The communication scheme consists of two phases: uplink training and data transmission. Uplink training phase consists of users transmitting training pilots, and base stations obtaining channel estimates. Data transmission phase consists of base stations transmitting data to the users through transmit precoding. Next, we describe these phases briefly and provide a set of achievable data rates using a given precoding method.

\subsection{Uplink Training}
At the beginning of every coherence interval, all users (in all cells) transmit training sequences, which are $\tau$ length column vectors. Let $\sqrt{\tau}\underline{\psi}_{jk}$ (normalized such that $\underline{\psi}_{jk}^{\dagger}\underline{\psi}_{jk}=1$) be the training vector transmitted by the $k$-th user in the $j$-th cell. Consider the base station of the $l$-th cell. The $\tau$ length column vector received at the $m$-th antenna of this base station is
\begin{equation}
\label{eq:ylm}
\underline{y}_{lm}=\sum_{j=1}^{L} \sum_{k=1}^{K} \sqrt{p_r \tau \beta_{jlk}} h_{jlkm} \underline{\psi}_{jk} + \underline{w}_{lm},
\end{equation}
where $\underline{w}_{lm}$ is the additive noise. Let $\mathbf{Y}_l=[\underline{y}_{l1} \: \underline{y}_{l2} \: \cdots \: \underline{y}_{lM}]$ ($\tau \times M$ matrix), 
$\mathbf{W}_l=[\underline{w}_{l1} \: \underline{w}_{l2} \: \cdots \: \underline{w}_{lM}]$ ($\tau \times M$ matrix), 
$\mathbf{\Psi}_j = [\underline{\psi}_{j1} \: \underline{\psi}_{j2} \: \cdots \: \underline{\psi}_{jK}]$ ($\tau \times K$ matrix), $\mathbf{D}_{jl}=\diag\{[\beta_{jl1} \: \beta_{jl2} \: \cdots \: \beta_{jlK}]\}$, and
\[ \mathbf{H}_{jl} = \left[ \begin{array}{ccc}
h_{jl11} & \cdots & h_{jl1M} \\
\vdots & \ddots & \vdots \\
h_{jlK1} & \cdots & h_{jlKM} \end{array} \right].\]
From (\ref{eq:ylm}), the signal received at this base station can be expressed as
\begin{equation}
\label{eq:Yl}
\mathbf{Y}_{l} = \sqrt{p_r \tau} \sum_{j=1}^{L} \left(\mathbf{\Psi}_j \mathbf{D}_{jl}^{\frac{1}{2}} \mathbf{H}_{jl}\right)+\mathbf{W}_{l}.
\end{equation}
The MMSE estimate of the channel $\mathbf{H}_{il}$ given $\mathbf{Y}_{l}$ in (\ref{eq:Yl}) is
\begin{equation}
\label{eq:mmse}
\mathbf{\hat{H}}_{jl}=\sqrt{p_r \tau } \mathbf{D}_{jl}^{\frac{1}{2}} \mathbf{\Psi}_j^{\dagger} \left(\mathbf{I}+p_r \tau \sum_{i=1}^{L} \mathbf{\Psi}_i \mathbf{D}_{il} \mathbf{\Psi}_i^{\dagger}\right)^{-1}\mathbf{Y}_{l}.
\end{equation}
This MMSE estimate in (\ref{eq:mmse}) follows from standard results in estimation theory (for example see \cite{KSH00}). We denote the MMSE estimate of the channel between this base station and all users by
$\hat{\mathbf{H}}_l = [\mathbf{\hat{H}}_{1l} \: \mathbf{\hat{H}}_{2l} \: \cdots \: \mathbf{\hat{H}}_{Ll}]$. This notation is used later in Section \ref{sec:pilot_cont}.

\subsection{Downlink Transmission}

Consider the base station of the $l$-th cell. Let the information symbols to be transmitted to users in the $l$-th cell be $\underline{q}_l = [q_{l1} \: q_{l2} \: \cdots \: q_{lK}]^T$ and the $M \times K$ linear precoding matrix be $\mathbf{A}_l=f(\hat{\mathbf{H}}_l)$. The function $f(\cdot)$ corresponds to the specific (linear) precoding method performed at the base station. The signal vector transmitted by this base station is $\mathbf{A}_l\underline{q}_l$. We consider transmission symbols and precoding methods such that $\mathbb{E}[\underline{q}_l]=\underline{0}$, $\mathbb{E}[\underline{q}_l\underline{q}_l^{\dagger}]=\mathbf{I}$, and $\tr\{\mathbf{A}_l^{\dagger}\mathbf{A}_l\}=1$. These (sufficient) conditions imply that the average power constraint at the base station is satisfied.

Now, consider the users in the $j$-th cell. The noisy signal vector received by these users is
\begin{equation}
\label{rcvd:vec}
\underline{x}_{j} = \sum_{l=1}^{L} \sqrt{p_f} \mathbf{D}_{jl}^{\frac{1}{2}}\mathbf{H}_{jl} \mathbf{A}_l\underline{q}_l + \underline{z}_{j}, \mbox{ ($K\times 1$ vector)}
\end{equation}
where $\underline{z}_{j}$ is the additive noise. From (\ref{rcvd:vec}), the signal received by the $k$-th user can be expressed as
\begin{equation}
\label{rcvd:signal}
{x}_{jk} = \sum_{l=1}^{L} \sum_{i=1}^{K} \sqrt{p_f \beta_{jlk}} [h_{jlk1} \: h_{jlk2} \: \cdots \: h_{jlkM}] \underline{a}_{li}{q}_{li} + {z}_{jk},
\end{equation}
where $\underline{a}_{li}$ is the $i$-th column of the precoding matrix $\mathbf{A}_l$ and ${z}_{jk}$ is the $k$-th element of $\underline{z}_{j}$.

\begin{remark}
The precoding method considered here is linear. This is of high practical value due to its low online complexity. Note that we do not consider non-linear precoding methods in this paper. We do not consider any training in the forward link as well. Therefore, the users do not have any channel knowledge. However, the pilot contamination problem is due to uplink training with non-orthogonal pilots, and hence, it will be present in all these settings.
\end{remark}
\subsection{Achievable Rates}
We provide a set of achievable rates using the above mentioned communication scheme by assuming worst-case Gaussian noise. This method was suggested in \cite{Marzetta:howmuchtraining} as a lower bounding technique. 
Equation (\ref{rcvd:signal}) can be written in the form
\begin{eqnarray}
{x}_{jk}& = &\sum_{l = 1}^{M} \sum_{i = 1}^{K} g^{jk}_{li}{q}_{li} + {z}_{jk}, \nonumber \\
\label{eq:xjk}
& = &\expec\left[{g}^{jk}_{jk}\right] {q}_{jk} + \left({g}^{jk}_{jk}-\expec\left[{g}^{jk}_{jk}\right]\right) {q}_{jk} + 
\sum_{(l,i)\ne (j,k)} g^{jk}_{li}{q}_{li} + {z}_{jk},
\end{eqnarray}
where $g^{jk}_{li}=\sqrt{p_f \beta_{jlk}} \left[h_{jlk1} \cdots h_{jlkM}\right] \underline{a}_{li}$. The effective noise in (\ref{eq:xjk}) is 
\begin{eqnarray}
\label{eq:effnoise}
{z}_{jk}^{'} = \left({g}^{jk}_{jk}-\expec\left[{g}^{jk}_{jk}\right]\right) {q}_{jk}+\sum_{(l,i)\ne (j,k)} g^{jk}_{li}{q}_{li} + {z}_{jk}.
\end{eqnarray}
 The random variables ${q}_{jk}$ and ${z}_{jk}^{'}$ given by (\ref{eq:effnoise}) are uncorrelated due to the following. ${q}_{jk}$ is clearly independent of ${q}_{li}$ for all $(l,i)\ne (j,k)$ and ${z}_{jk}$. In addition, ${q}_{jk}$ is independent of ${g}^{jk}_{jk}$. Therefore,
\begin{eqnarray}
\expec\left[\left({g}^{jk}_{jk}-\expec\left[{g}^{jk}_{jk}\right]\right) |{q}_{jk}|^2\right]=\expec\left[\left({g}^{jk}_{jk}-\expec\left[{g}^{jk}_{jk}\right]\right)\right] \expec\left[|{q}_{jk}|^2\right]=0.
\end{eqnarray} The variance of the effective noise is $$\expec\left[|{z}_{jk}^{'}|^2\right]=\var\left\{|{g}^{jk}_{jk}|\right\}+\sum_{(l,i)\ne (j,k)}\expec\left[|g^{jk}_{li}|^2\right]+1.$$
 Using the fact that worst-case uncorrelated noise distribution is independent Gaussian noise \cite{Hassibi:howmuchtraining} with the same variance, we obtain the following set of achievable rates: 
\begin{equation}
\label{eq:achrate}
R_{jk}= C\left(\frac{\left|\expec\left[{g}^{jk}_{jk}\right]\right|^2}{1+\var\left\{\left|{g}^{jk}_{jk}\right|\right\}+\sum_{(l,i)\ne (j,k)}\expec\left[\left|g^{jk}_{li}\right|^2\right]}\right),
\end{equation}
where $C(\theta)=\log_2(1+\theta)$.

\begin{remark}
The set of achievable rates given by (\ref{eq:achrate}) is valid for any linear precoding method, and depends on the precoding method through the expectation and variance terms appearing in (\ref{eq:achrate}). 
\end{remark}

Similar achievable rates are used in the single-cell setting as well to study and/or compare precoding methods. Next, we perform pilot contamination analysis for zero-forcing precoding.

\section{Pilot Contamination Analysis}
\label{sec:pilot_cont}

We analyze the pilot contamination problem in the following setting: one user per cell ($K=1$), same training sequence used by all users ($\underline{\psi}_{j1}=\underline{\psi},\forall j$) and zero-forcing (ZF) precoding. We consider this setting as it captures the primary effect of pilot contamination which is the correlation between the precoding matrix (vector in this setting) used by the base station in a cell and channel to users in other cells. We provide simple and insightful analytical results in this setting. As mentioned earlier, we emphasize that the pilot contamination problem results from uplink training with non-orthogonal training sequences, and hence, it is not specific to the setting considered here. However, the level of its impact on the achievable rates would vary depending on the system settings.

In order to simplify notation, we drop the subscripts associated with the users in every cell. In this section, $\mathbf{H}_{jl}$, $\mathbf{\hat{H}}_{jl}$ and $\mathbf{A}_l$ are vectors and we denote these using $\underline{h}_{jl}$, $\underline{\hat{h}}_{jl}$ and $\underline{a}_l$, respectively. For zero-forcing precoding, the precoding vector used at the base station in the $l$-th cell is given by $\underline{a}_l =\underline{\hat{h}}_{ll}^{\dagger}/\|\sqrt{\underline{\hat{h}}_{ll}}\|.$
The user in the $j$-th cell receives signal from its base station and from other base stations. From (\ref{rcvd:vec}), this received signal is
\begin{equation}
\label{rcvd:xj}
{x}_{jk} = \sqrt{p_f \beta_{jj}} \underline{h}_{jj} \underline{a}_{j} {q}_{j} + \sum_{l \ne j} \sqrt{p_f \beta_{jl}} \underline{h}_{jl} \underline{a}_{l} {q}_{l} + {z}_{j}.
\end{equation}
We compute first and second order moments of the effective channel gain and the inter-cell interference and use these to obtain a simple expression for the achievable rate given by (\ref{eq:achrate}).

In the setting considered here, the MMSE estimate of $\underline{h}_{jl}$ based on $\mathbf{Y}_{l}$ given by (\ref{eq:mmse}) can be simplified using matrix inversion lemma and the fact that $\underline{\psi}^{\dagger}\underline{\psi}=1$ as follows:
\begin{eqnarray}
\underline{\hat{h}}_{jl}&=&\sqrt{p_r \tau \beta_{jl}} \underline{\psi}^{\dagger} \left(\mathbf{I}+\underline{\psi}\left(p_r \tau \sum_{i=1}^{L} \beta_{il}\right) \underline{\psi}^{\dagger}\right)^{-1}\mathbf{Y}_{l}, \nonumber \\
&=&\sqrt{p_r \tau \beta_{jl}} \underline{\psi}^{\dagger} \left(\mathbf{I}-\frac{\underline{\psi}\left(p_r \tau \sum_{i=1}^{L} \beta_{il}\right) \underline{\psi}^{\dagger}}{1+p_r \tau \sum_{i=1}^{L} \beta_{il}}\right)\mathbf{Y}_{l}, \nonumber \\
&=&\sqrt{p_r \tau \beta_{jl}}  \left(\underline{\psi}^{\dagger}-\frac{\left(p_r \tau \sum_{i=1}^{L} \beta_{il}\right) }{1+p_r \tau \sum_{i=1}^{L} \beta_{il}}\underline{\psi}^{\dagger}\right)\mathbf{Y}_{l}, \nonumber \\
&=&\frac{\sqrt{p_r \tau \beta_{jl}}}{1+ p_r \tau \sum_{i=1}^{L} \beta_{il}} \underline{\psi}^{\dagger}\mathbf{Y}_{l}. \nonumber
\end{eqnarray}
Since $\underline{\psi}^{\dagger}\mathbf{Y}_{l}$ is proportional to the MMSE estimate of $\underline{h}_{jl}$ for any $j$, we have
\begin{equation}
\label{eq:ratio}
\frac{\underline{\hat{h}}_{jl}^{\dagger}}{\|\underline{\hat{h}}_{jl}\|}=\frac{\mathbf{Y}_{l}^{\dagger}\underline{\psi}}{\|\underline{\psi}^{\dagger}\mathbf{Y}_{l}\|}, \forall j.
\end{equation}
Using (\ref{eq:ratio}), we obtain
\begin{eqnarray}
\underline{h}_{jl} \underline{a}_{l} & = & \underline{h}_{jl} \frac{\underline{\hat{h}}_{jl}^{\dagger}}{\left\|\underline{\hat{h}}_{jl}^{\dagger}\right\|}, \nonumber \\
\label{eq:ha}
& = & \left\|\underline{\hat{h}}_{jl}^{\dagger}\right\| + \underline{\tilde{h}}_{jl} \frac{\underline{\hat{h}}_{jl}^{\dagger}}{\left\|\underline{\hat{h}}_{jl}^{\dagger}\right\|},
\end{eqnarray}
where $\underline{\tilde{h}}_{jl} = \underline{h}_{jl}-\underline{\hat{h}}_{jl}$. From the properties of MMSE estimation, we know that $\underline{\hat{h}}_{jl}$ is independent of $\underline{\tilde{h}}_{jl}$, $\underline{\hat{h}}_{jl}$ is $\mathcal{CN}\left(\underline{0},\frac{p_r \tau \beta_{jl}}{1+ p_r \tau \sum_{i=1}^{L} \beta_{il}}\mathbf{I}\right)$, and $\underline{\tilde{h}}_{jl}$ is $\mathcal{CN}\left(\underline{0},\frac{1+ p_r \tau \sum_{i\ne j} \beta_{il}}{1+ p_r \tau \sum_{i=1}^{L} \beta_{il}}\mathbf{I}\right)$. These results are used next.

From (\ref{eq:ha}), we get
\begin{eqnarray}
\expec\left[\underline{h}_{jl} \underline{a}_{l}\right] & = & \expec\left[\left\|\underline{\hat{h}}_{jl}^{\dagger}\right\|\right], \nonumber \\
\label{eq:haexpec}
& = & \sqrt{\frac{p_r \tau \beta_{jl}}{1+ p_r \tau \sum_{i=1}^{L} \beta_{il}}}\expec\left[\theta\right],
\end{eqnarray}
where $\theta = \sqrt{\sum_{m=1}^{M}|u_m|^2}$ and $\{u_m\}$ is i.i.d. $\mathcal{CN}(0,1)$. From (\ref{eq:ha}), we also have
\begin{eqnarray}
\expec \left[\left\|\underline{h}_{jl} \underline{a}_{l}\right\|^2\right] & = & \expec\left[\left\|\underline{\hat{h}}_{jl}^{\dagger}\right\|^2 \right] + \expec \left[\frac{\underline{\hat{h}}_{jl}}{\left\|\underline{\hat{h}}_{jl}^{\dagger}\right\|}\underline{\tilde{h}}_{jl}^{\dagger} \underline{\tilde{h}}_{jl} \frac{\underline{\hat{h}}_{jl}^{\dagger}}{\left\|\underline{\hat{h}}_{jl}^{\dagger}\right\|}\right], \nonumber \\
\label{eq:ha2expec}
& = & \frac{p_r \tau \beta_{jl}}{1+ p_r \tau \sum_{i=1}^{L} \beta_{il}}\expec[\theta^2] + \frac{1+ p_r \tau \sum_{i\ne j} \beta_{il}}{1+ p_r \tau \sum_{i=1}^{L} \beta_{il}}.
\end{eqnarray}

Next, we state two lemmas required to obtain a closed-form expression for the achievable rate.
\begin{lemma}
\label{lemma:effch}
The effective channel gain in (\ref{rcvd:xj}) has expectation 
\[
\expec\left[\sqrt{p_f \beta_{jj}} \underline{h}_{jj} \underline{a}_{j} \right]= \left(p_f \beta_{jj}\frac{p_r \tau \beta_{jj}}{1+ p_r \tau \sum_{i=1}^{L} \beta_{ij}}\right)^{1/2}\expec[\theta]\]
and variance 
\[\var\left\{\left|\sqrt{p_f \beta_{jj}} \underline{h}_{jj} \underline{a}_{j} \right|\right\}=p_f \beta_{jj}\left(\frac{p_r \tau \beta_{jj}}{1+ p_r \tau \sum_{i=1}^{L} \beta_{ij}}\var\{\theta\} + \frac{1+ p_r \tau \sum_{i\ne j} \beta_{ij}}{1+ p_r \tau \sum_{i=1}^{L} \beta_{ij}}\right).
\]
\end{lemma}
\begin{IEEEproof}
The proof follows from (\ref{eq:haexpec}) and (\ref{eq:ha2expec}). Note that $\var\{\theta\} =\expec[\theta^2]-(\expec[\theta])^2$ by definition.
\end{IEEEproof}

\begin{lemma}
\label{lemma:moments}
For both signal and interference terms in (\ref{rcvd:xj}), the first and second order moments are as follows:
\[
\expec\left[\sqrt{p_f \beta_{jl}} \underline{h}_{jl} \underline{a}_{l} {q}_{l}\right]=0,
\]
\[
\expec\left[\left|\sqrt{p_f \beta_{jl}} \underline{h}_{jl} \underline{a}_{l} {q}_{l}\right|^2\right]=p_f \beta_{jl}\left(\frac{p_r \tau \beta_{jl}}{1+ p_r \tau \sum_{i=1}^{L} \beta_{il}}\expec[\theta^2] + \frac{1+ p_r \tau \sum_{i\ne j} \beta_{il}}{1+ p_r \tau \sum_{i=1}^{L} \beta_{il}}\right).
\]
\end{lemma}
\begin{IEEEproof}
Since $\expec[{q}_{l}]=0$ and ${q}_{l}$ is independent of $\underline{h}_{jl}$ and $\underline{a}_{l}$, it is clear that 
\[
\expec\left[\sqrt{p_f \beta_{jl}} \underline{h}_{jl} \underline{a}_{l} {q}_{l}\right]=0.
\] The proof of the second order moment follows directly from (\ref{eq:ha2expec}).
\end{IEEEproof}

The main result of this section is given in the next theorem. This theorem provides a closed-form expression for the achievable rates under the setting considered in this section, i.e., one user per cell ($K=1$), same training sequence used by all users ($\underline{\psi}_{j1}=\underline{\psi},\forall j$) and zero-forcing (ZF) precoding.

\begin{theorem}
\label{thm:achrate_sim}
For the setting considered, the achievable rate of the user in the $j$-th cell during downlink transmission in (\ref{eq:achrate}) is given by
\begin{equation}
\label{eq:Rj}
R_j = C\left(\frac{p_f \beta_{jj}\frac{p_r \tau \beta_{jj}}{\kappa_j}\expec^2[\theta]}{1+p_f \beta_{jj}\frac{p_r \tau \beta_{jj}}{\kappa_j}\var\{\theta\}+\sum_{l \ne j} p_f \beta_{jl}\frac{p_r \tau \beta_{jl}}{\kappa_l}\expec[\theta^2] + \sum_{l=1}^{L} p_f \beta_{jl} \frac{1+ p_r \tau \sum_{i\ne j} \beta_{il}}{\kappa_l}}\right),
\end{equation}
where $\kappa_j = 1+ p_r \tau \sum_{i=1}^{L} \beta_{ij}$, $\expec[\theta]=\frac{\Gamma(M+\frac{1}{2})}{\Gamma(M)}$, $\expec[\theta^2]=M$ and $\var\{\theta\}=M-\expec^2[\theta]$. Here, $\Gamma(\cdot)$ is the Gamma function. For large M, the following limiting expression for achievable rate can be obtained:
\begin{equation}
\label{eq:limRj}
\lim_{M \rightarrow \infty} R_j = C\left(\frac{\frac{\beta_{jj}^2}{1+ p_r \tau \sum_{i=1}^{L} \beta_{ij}}}{\sum_{l \ne j} \frac{\beta_{jl}^2}{1+ p_r \tau \sum_{i=1}^{L} \beta_{il}}}\right).
\end{equation}
\end{theorem}

\begin{IEEEproof}
The proof of (\ref{eq:Rj}) follows by substituting the results of Lemma \ref{lemma:effch} and Lemma \ref{lemma:moments} in (\ref{eq:achrate}). Since $\theta$ has a scaled (by a factor of $1/\sqrt{2}$) chi distribution with $2M$ degrees of freedom, it is straightforward to see that $\expec[\theta]=\frac{\Gamma(M+\frac{1}{2})}{\Gamma(M)}$, $\expec[\theta^2]=M$ and $\var\{\theta\}=M-\expec^2[\theta]$.

Using the duplication formula 
\[
\Gamma(z)\Gamma\left(z+\frac{1}{2}\right)=2^{(1-2z)}\sqrt{\pi}\Gamma(2z)
\] and Stirling's formula 
\[
\lim_{n \rightarrow \infty} \frac{n!}{\sqrt{2\pi n}n^n e^{-n}} = 1,
\] we obtain
\begin{eqnarray}
\lim_{M \rightarrow \infty} \frac{1}{\sqrt{M}}\frac{\Gamma\left(M+\frac{1}{2}\right)}{\Gamma(M)} & = & \lim_{M \rightarrow \infty} \sqrt{\frac{\pi}{{M}}}2^{(1-2M)}\frac{(2M-1)!}{(M-1)!(M-1)!}, \nonumber \\
& = & \lim_{M \rightarrow \infty} \sqrt{\frac{\pi}{{M}}}2^{(1-2M)} \frac{\sqrt{2\pi(2M-1)} (2M-1)^{(2M-1)}e^{1-2M}}{2\pi(M-1) (M-1)^{2(M-1)}e^{2(1-M)}}, \nonumber \\
& = & \lim_{M \rightarrow \infty} \sqrt{\frac{2M-1}{{2M}}}\left(1+\frac{1}{2(M-1)}\right)^{2M-1} e^{-1}, \nonumber \\
& = & 1. \nonumber
\end{eqnarray}
Therefore, $\lim_{M \rightarrow \infty} \frac{\expec^2[\theta]}{M} = 1$ and $\lim_{M \rightarrow \infty} \frac{\var\{\theta\}}{M} = 0$. This completes the proof of (\ref{eq:limRj}).
\end{IEEEproof}

For large $M$, the value of $\var\{\theta\}$ ($\approx 1/4$) is insignificant compared to $M$. The results of the above theorem show that the performance does saturate with $M$. Typically, the reverse link is interference-limited, i.e., $p_r \tau \sum_{i=1}^{L} \beta_{il} \gg 1, \forall j$. The term $\sum_{i=1}^{L} \beta_{il}$ is the expected sum of squares of the propagation coefficients between the base station in the $j$-th cell and all users. Therefore, $\sum_{i=1}^{L} \beta_{il}$ is generally constant with respect to $j$. Using these approximations in (\ref{eq:limRj}), we get $$ R_j \approx C\left(\frac{\beta_{jj}^2}{\sum_{l \ne j} \beta_{jl}^2} \right).$$ This clearly show that the impact of pilot contamination can be very significant if cross gains (between cells) are of the same order of direct gains (within the same cell). It suggests frequency/time reuse and pilot reuse techniques to reduce the cross gains (in the same frequency/time) relative to the direct gains. The benefits of frequency reuse in the limit of an infinite number of antennas were demonstrated in \cite{Marzetta09}.

\begin{remark}
Our result in Theorem \ref{thm:achrate_sim} is not an asymptotic result. The expression in (\ref{eq:Rj}) is exact for any value of the number of antennas $M$ at the base stations. Hence, this expression can be used to find the appropriate frequency/time reuse scheme for any given value of $M$ and other system parameters. We do not focus on this in this paper, as this would depend largely on the actual system parameters including the cell layout and the shadowing model.
\end{remark}

\begin{remark}
The result in Theorem \ref{thm:achrate_sim} is for the setting with one user per cell. In the general setting with $K$ users per cell, the analysis in this section does not directly extend to provide a closed-form expression. However, (\ref{eq:achrate}) can be numerically evaluated in the general setting.
\end{remark}

To summarize, the impact of uplink training with non-orthogonal pilots can be serious when the cross-gains are not small compared to the direct gains. This pilot contamination problem is often neglected in theory and even in many large-scale simulations. The analysis in this section shows the need to account for this impact especially in systems with high reuse of training sequences. In addition to uplink training in TDD systems, which is the focus of this paper, the pilot contamination problem would appear in other scenarios as well as it is fundamental to training with non-orthogonal pilots. 

Next, we proceed to develop a new precoding method referred to as the multi-cell MMSE-based precoding in this paper.

\section{Multi-Cell MMSE-Based Precoding}
\label{sec:mmseprecoding}

In the previous section, we show that pilot contamination severely impacts the system performance by increasing the inter-cell interference. In particular, we show that the inter-cell interference grows like the intended signal with the number of antennas $M$ at the base stations while using zero-forcing precoding. Therefore, in the presence of pilot contamination, in addition to frequency/time/pilot reuse schemes, it is crucial to account for inter-cell interference while designing a precoding method. Furthermore, since pilot contamination is originating from the non-orthogonal training sequences, it is important to account for the training sequence allocation while designing a precoding method. The approach of accounting for inter-cell interference while designing a precoding method is common, while the approach of accounting for the training sequence allocation is not. Again, the usual approach is to decouple the channel estimation and precoding completely. However, while using non-orthogonal pilots, this is not the right approach. These observations follow from our pilot contamination analysis in the previous section.

The precoding problem cannot be directly formulated as a joint optimization problem as different base stations have different received training signals. In other words, the problem is decentralized in nature. Therefore, one approach is to apply single-cell precoding methods. For example, since we assume orthogonal training sequences in every cell, we can perform zero-forcing on the users in every cell. The precoding matrix corresponding to this zero-forcing approach is given by
\begin{equation}
\label{eq:ZF}
\mathbf{A}_{l}=\frac{\mathbf{\hat{G}}^{\dagger}_{ll}\left(\mathbf{\hat{G}}_{ll}\mathbf{\hat{G}}^{\dagger}_{ll}\right)^{-1}}{\sqrt{\tr\left[\left(\mathbf{\hat{G}}_{ll}\mathbf{\hat{G}}^{\dagger}_{ll}\right)^{-1}\right]}},
\end{equation}
where $\mathbf{\hat{G}}_{ll} = \sqrt{p_f}\mathbf{D}^{\frac{1}{2}}_{ll}\mathbf{\hat{H}}_{ll}$.
However, this zero-forcing precoding or other single-cell precoding methods do not account for the training sequence allocation, which is potentially the right approach to mitigate the pilot contamination problem. We explore this next.

In order to determine the precoding matrices, we formulate an optimization problem for each precoding matrix. Consider the $j$-th cell. The signal received by the users in this cell given by (\ref{rcvd:vec}) is a function of all the precoding matrices (used at all the base stations). Therefore, the MMSE-based precoding methods for single-cell setting considered in \cite{GPS08} does not extend (directly) to this setting. Let us consider the signal and interference terms corresponding to the base station in the $l$-th cell. Based on these terms, we formulate the following optimization problem to obtain the precoding matrix $\mathbf{A}_{l}$. We use the following notation: $\mathbf{{F}}_{jl} = \sqrt{p_f}\mathbf{D}^{\frac{1}{2}}_{jl}\mathbf{{H}}_{jl}$, $\mathbf{\hat{F}}_{jl} = \sqrt{p_f}\mathbf{D}^{\frac{1}{2}}_{jl}\mathbf{\hat{H}}_{jl}$ and $\mathbf{\tilde{F}}_{jl} = \mathbf{{F}}_{jl} - \mathbf{\hat{F}}_{jl}$ for all $j$ and $l$. The optimization problem is:
\begin{eqnarray}
\label{eq:MMSEform}
&\text{minimize}_{\mathbf{A}_{l},\alpha} & \expec_{\mathbf{\tilde{F}}_{jl},\underline{z}_l,\underline{q}_l}\left[\left\|\alpha(\mathbf{{F}}_{ll}\mathbf{A}_{l}\underline{q}_l+\underline{z}_l)-\underline{q}_l\right\|^2 + \sum_{j\ne l} \left\|\alpha\gamma(\mathbf{{F}}_{jl}\mathbf{A}_{l}\underline{q}_l)\right\|^2 \right]  \\
&\text{subject to} &\tr\{\mathbf{A}_{l}^{\dagger}\mathbf{A}_{l}\}=1. \nonumber
\end{eqnarray}
This objective function is very intuitive. The objective function of the problem (\ref{eq:MMSEform}) consists of two parts: (\emph{i}) the sum of squares of ``errors'' seen by the users in the $l$-th cell, and (\emph{ii}) the sum of squares of interference by the users in all other cells. The parameter $\gamma$ of the optimization problem ``controls'' the relative weights associated with these two parts. The parameter $\alpha$ is important as it ``virtually'' corresponds to the potential scaling that can be performed at the users. The optimal solution to the problem (\ref{eq:MMSEform}) denoted by $\mathbf{A}_{l}^{opt}$ is the multi-cell MMSE-based precoding matrix.

Next, we obtain a closed-form expression for $\mathbf{A}_{l}^{opt}$. The following lemma is required later for obtaining the optimal solution to the problem (\ref{eq:MMSEform}).

\begin{lemma}
\label{lemma:deltajl}
Consider the optimization problem (\ref{eq:MMSEform}). For all $j$ and $l$, $$\expec\left[\mathbf{\tilde{F}}_{jl}^{\dagger}\mathbf{\tilde{F}}_{jl}\right] = \delta_{jl} \mathbf{I}_M,$$ where
\begin{equation}
\label{eq:deltajl}
\delta_{jl} = p_f\tr \left\{\mathbf{D}_{jl}\left(\mathbf{I}_K + p_r \tau \mathbf{D}_{jl}^{\frac{1}{2}} \mathbf{\Psi}_j^{\dagger} \mathbf{\Lambda}_{jl}\mathbf{\Psi}_j\mathbf{D}_{jl}^{\frac{1}{2}}\right)^{-1}\right\},
\end{equation}
and 
\[\mathbf{\Lambda}_{jl}=\left(\mathbf{I}+p_r \tau \sum_{i \ne j} \mathbf{\Psi}_i \mathbf{D}_{il} \mathbf{\Psi}_i^{\dagger}\right)^{-1}.
\]
\end{lemma}

\begin{IEEEproof}
Let $\underline{\tilde{f}}_{jlm}$ denote the $m$-th column of $\mathbf{\tilde{F}}_{jl}$. Similarly, we define $\underline{{h}}_{jlm}$ and $\underline{\hat{h}}_{jlm}$. From (\ref{eq:mmse}), we have
\begin{eqnarray}
\underline{\tilde{f}}_{jlm} &= &\sqrt{p_f}\mathbf{D}^{\frac{1}{2}}_{jl}(\underline{{h}}_{jlm} - \underline{\hat{h}}_{jlm}), \nonumber \\
&=& \sqrt{p_f}\mathbf{D}^{\frac{1}{2}}_{jl}\left(\underline{{h}}_{jlm} - \sqrt{p_r \tau } \mathbf{D}_{jl}^{\frac{1}{2}} \mathbf{\Psi}_j^{\dagger} \left(\mathbf{I}+p_r \tau \sum_{i=1}^{L} \mathbf{\Psi}_i \mathbf{D}_{il} \mathbf{\Psi}_i^{\dagger}\right)^{-1} \underline{y}_{lm} \right), \nonumber
\end{eqnarray}
where $\underline{y}_{lm}$ is given by (\ref{eq:ylm}). For given $j$ and $l$, it is clear that $\left\{\underline{\tilde{f}}_{jlm}\right\}_{m=1}^{M}$ is i.i.d. zero-mean $\mathcal{CN}$ distributed. Hence, $\expec\left[\mathbf{\tilde{F}}_{jl}^{\dagger}\mathbf{\tilde{F}}_{jl}\right] = \delta_{jl} \mathbf{I}_M$ where
\begin{eqnarray}
\delta_{jl} &= &\expec\left[\underline{\tilde{f}}_{jlm}^{\dagger}\underline{\tilde{f}}_{jlm}\right], \nonumber \\
&=& p_f\tr \left\{\mathbf{D}^{\frac{1}{2}}_{jl}\left(\mathbf{I}_K - \expec[\underline{\hat{h}}_{jlm}\underline{\hat{h}}_{jlm}^{\dagger}]\right)\mathbf{D}^{\frac{1}{2}}_{jl}\right\}, \nonumber \\
&=& p_f\tr \left\{\mathbf{D}^{\frac{1}{2}}_{jl}\left(\mathbf{I}_K - p_r \tau \mathbf{D}_{jl}^{\frac{1}{2}} \mathbf{\Psi}_j^{\dagger} \left(\mathbf{I}+p_r \tau \sum_{i=1}^{L} \mathbf{\Psi}_i \mathbf{D}_{il} \mathbf{\Psi}_i^{\dagger}\right)^{-1}\mathbf{\Psi}_j\mathbf{D}_{jl}^{\frac{1}{2}}\right)\mathbf{D}^{\frac{1}{2}}_{jl}\right\}, \nonumber \\
&=& p_f\tr \left\{\mathbf{D}^{\frac{1}{2}}_{jl}\left(\mathbf{I}_K + p_r \tau \mathbf{D}_{jl}^{\frac{1}{2}} \mathbf{\Psi}_j^{\dagger} \left(\mathbf{I}+p_r \tau \sum_{i \ne j} \mathbf{\Psi}_i \mathbf{D}_{il} \mathbf{\Psi}_i^{\dagger}\right)^{-1}\mathbf{\Psi}_j\mathbf{D}_{jl}^{\frac{1}{2}}\right)^{-1}\mathbf{D}^{\frac{1}{2}}_{jl}\right\}. \nonumber
\end{eqnarray}
The last step follows from matrix inversion lemma. This completes the proof of the lemma.
\end{IEEEproof}

The main result of this section is given by the following theorem. This theorem provides a closed-form expression for the multi-cell MMSE-based precoding matrix. 

\begin{theorem}
\label{th:Aopt}
The optimal solution to the problem (\ref{eq:MMSEform}) is
\begin{equation}
\label{eq:Aoptthm}
\mathbf{A}_{l}^{opt} = \frac{1}{\alpha^{opt}}\left(\mathbf{\hat{F}}_{ll}^{\dagger}\mathbf{\hat{F}}_{ll}+\gamma^2\sum_{j\ne l} \mathbf{\hat{F}}_{jl}^{\dagger}\mathbf{\hat{F}}_{jl}+\eta \mathbf{I}_M\right)^{-1}\mathbf{\hat{F}}_{ll}^{\dagger},
\end{equation}
where $$\eta=\delta_{ll}+\gamma^2\sum_{j\ne l}\delta_{jl}+K,$$ $\delta_{jl}$ is given by (\ref{eq:deltajl}) and $\alpha^{opt}$ is such that $\tr\left\{\left(\mathbf{A}_{l}^{opt}\right)^{\dagger} \mathbf{A}_{l}^{opt}\right\}=1.$
\end{theorem}

\begin{proof}
First, we simplify the objective function $J(\mathbf{A}_{l},\alpha)$ of the problem (\ref{eq:MMSEform}) as follows:
\begin{eqnarray}
J(\mathbf{A}_{l},\alpha) & = & \expec \left[\left\|\alpha\left(\mathbf{{F}}_{ll}\mathbf{A}_{l}\underline{q}_l+\underline{z}_l\right)-\underline{q}_l\right\|^2 + \sum_{j\ne l} \left\|\alpha\gamma\mathbf{{F}}_{jl}\mathbf{A}_{l}\underline{q}_l\right\|^2 \right], \nonumber \\
& = & \expec \left[\left\|\left(\alpha\mathbf{{F}}_{ll}\mathbf{A}_{l}-\mathbf{I}_K\right)\underline{q}_l\right\|^2 + \sum_{j\ne l} \left\|\alpha\gamma\mathbf{{F}}_{jl}\mathbf{A}_{l}\underline{q}_l\right\|^2 \right]+ \alpha^2 K, \nonumber \\
& = & \tr \left\{\expec \left[\left(\alpha\mathbf{{F}}_{ll}\mathbf{A}_{l}-\mathbf{I}_K\right)^{\dagger}\left(\alpha\mathbf{{F}}_{ll}\mathbf{A}_{l}-\mathbf{I}_K\right) + \sum_{j\ne l} \alpha^2\gamma^2\mathbf{A}_{l}^{\dagger}\mathbf{{F}}_{jl}^{\dagger}\mathbf{{F}}_{jl}\mathbf{A}_{l} \right]\right\}+ \alpha^2 K, \nonumber \\
& = & \tr \left\{\alpha^2\mathbf{A}_{l}^{\dagger}\expec \left[\mathbf{{F}}_{ll}^{\dagger}\mathbf{{F}}_{ll}\right]\mathbf{A}_{l} + \sum_{j\ne l} \alpha^2\gamma^2\mathbf{A}_{l}^{\dagger}\expec \left[\mathbf{{F}}_{jl}^{\dagger}\mathbf{{F}}_{jl}\right]\mathbf{A}_{l} - \alpha\mathbf{A}_{l}^{\dagger}\mathbf{\hat{F}}_{ll}^{\dagger} -\alpha\mathbf{\hat{F}}_{ll}\mathbf{A}_{l}  \right\}+ \left(\alpha^2+1\right) K,\nonumber \\
& = & \tr \left\{\alpha^2\mathbf{A}_{l}^{\dagger}\left(\mathbf{\hat{F}}_{ll}^{\dagger}\mathbf{\hat{F}}_{ll}+\gamma^2\sum_{j\ne l} \mathbf{\hat{F}}_{jl}^{\dagger}\mathbf{\hat{F}}_{jl}+\left(\delta_{ll}+\gamma^2\sum_{j\ne l}\delta_{jl}\right)\mathbf{I}_M\right)\mathbf{A}_{l} - \alpha\mathbf{A}_{l}^{\dagger}\mathbf{\hat{F}}_{ll}^{\dagger} -\alpha\mathbf{\hat{F}}_{ll}\mathbf{A}_{l}  \right\}+ \nonumber \\
& & \left(\alpha^2+1\right) K. \nonumber
\end{eqnarray}
The last step follows from Lemma \ref{lemma:deltajl}. 

Now, consider the Lagrangian formulation
$$L\left(\mathbf{A}_{l},\alpha,\lambda\right) = J\left(\mathbf{A}_{l},\alpha\right) + \lambda\left(\tr\left\{\mathbf{A}_{l}^{\dagger}\mathbf{A}_{l}\right\}-1\right)$$ for the problem (\ref{eq:MMSEform}). Let 
\[
\mathbf{R}=\mathbf{\hat{F}}_{ll}^{\dagger}\mathbf{\hat{F}}_{ll}+\gamma^2\sum_{j\ne l} \mathbf{\hat{F}}_{jl}^{\dagger}\mathbf{\hat{F}}_{jl}+\left(\delta_{ll}+\gamma^2\sum_{j\ne l}\delta_{jl}+\frac{\lambda}{\alpha^2}\right)\mathbf{I}_M,
\] $\mathbf{U}=\alpha\mathbf{R}^{\frac{1}{2}}\mathbf{A}_{l}$ and $\mathbf{V}=\mathbf{R}^{-\frac{1}{2}}\mathbf{\hat{F}}_{ll}^{\dagger}$. We have
\begin{equation}
\label{eq:lagr}
L\left(\mathbf{A}_{l},\alpha,\lambda\right) = \left\|\mathbf{U}- \mathbf{V}\right\|^2 - \tr\left\{\mathbf{\hat{F}}_{ll}\mathbf{R}^{-1}\mathbf{\hat{F}}_{ll}^{\dagger} \right\}+\left(\alpha^2+1\right) K-\lambda.
\end{equation}
This can be easily verified by expanding the right hand side. It is clear from (\ref{eq:lagr}) that, for any given $\alpha$ and $\lambda$, $L(\mathbf{A}_{l},\alpha,\lambda)$ is minimized if and only if $\mathbf{U} = \mathbf{V}$. Hence, we obtain
\begin{equation}
\label{eq:Aoptfirst}
\mathbf{A}_{l}^{opt} =\frac{1}{\alpha}\mathbf{R}^{-1}\mathbf{\hat{F}}_{ll}^{\dagger}.
\end{equation}

Let $L(\alpha,\lambda) = L(\mathbf{A}_{l}^{opt},\alpha,\lambda)$. Now, we have
\begin{equation}
\label{eq:L}
L(\alpha,\lambda) = - \tr\left\{\mathbf{\hat{F}}_{ll}\mathbf{R}^{-1}\mathbf{\hat{F}}_{ll}^{\dagger} \right\}+\left(\alpha^2+1\right) K-\lambda.
\end{equation}
Note that $\mathbf{\hat{F}}_{ll}^{\dagger}\mathbf{\hat{F}}_{ll}+\gamma^2\sum_{j\ne l} \mathbf{\hat{F}}_{jl}^{\dagger}\mathbf{\hat{F}}_{jl}$ can be factorized in the form $\mathbf{S}^{\dagger}\diag\{[c_1 \: c_2 \: \cdots \: c_M]\}\mathbf{S}$ where $\mathbf{S}^{\dagger}\mathbf{S}=\mathbf{I}_{M}.$ Let $\delta = \delta_{ll}+ \gamma^2\sum_{j\ne l}\delta_{jl}$. Therefore,

\begin{eqnarray}
\mathbf{R}^{-1} & = &\left(\mathbf{S}^{\dagger}\diag\{[c_1 \: c_2 \: \cdots \: c_M]\}\mathbf{S}+\left(\delta+\frac{\lambda}{\alpha^2}\right)\mathbf{I}_M\right)^{-1} \nonumber, \\
& = &\left(\mathbf{S}^{\dagger}\diag\left\{\left[c_1+\delta+\frac{\lambda}{\alpha^2} \: c_2+\delta+\frac{\lambda}{\alpha^2} \: \cdots \: c_M+\delta+\frac{\lambda}{\alpha^2}\right]\right\}\mathbf{S}\right)^{-1}, \nonumber \\
\label{eq:Rinv}
& = &\mathbf{S}^{\dagger}\diag\left\{\left[\left(c_1+\delta+\frac{\lambda}{\alpha^2}\right)^{-1} \: \left(c_2+\delta+\frac{\lambda}{\alpha^2}\right)^{-1} \: \cdots \: \left(c_M+\delta+\frac{\lambda}{\alpha^2}\right)^{-1}\right]\right\}\mathbf{S}.
\end{eqnarray}
Substituting (\ref{eq:Rinv}) in (\ref{eq:L}), we get
\begin{equation}
\label{eq:Lsimp}
L(\alpha,\lambda) = -\sum_{m=1}^{M} \frac{d_m}{c_m+\delta+\frac{\lambda}{\alpha^2}}+(\alpha^2+1) K-\lambda,
\end{equation}
where $d_m$ is the $(m,m)$-th entry of $\mathbf{S}\mathbf{\hat{F}}_{ll}^{\dagger}\mathbf{\hat{F}}_{ll}\mathbf{S}^{\dagger}.$
Consider the equations obtained by differentiating (\ref{eq:Lsimp}) w.r.t. $\alpha$ and $\lambda$ and equating to zero:
\begin{equation}
\label{eq:difflamda}
\sum_{m=1}^{M} \frac{d_m}{\left(c_m+\delta+\frac{\lambda}{\alpha^2}\right)^2}\frac{1}{\alpha^2}=1,
\end{equation}
\begin{equation}
\label{eq:diffalpha}
-\sum_{m=1}^{M} \frac{d_m}{\left(c_m+\delta+\frac{\lambda}{\alpha^2}\right)^2}\frac{2\lambda}{\alpha^3}+2 \alpha K=0.
\end{equation}
Substituting (\ref{eq:difflamda}) in (\ref{eq:diffalpha}), we get 
\begin{equation}
\label{eq:lamdaalpha}
\frac{\lambda}{\alpha^2} = K.
\end{equation}
Combining the results in (\ref{eq:Aoptfirst}), (\ref{eq:Rinv}), and (\ref{eq:lamdaalpha}), we have
$$\mathbf{A}_{l}^{opt} =\frac{1}{\alpha^{opt}}\left(\mathbf{\hat{F}}_{ll}^{\dagger}\mathbf{\hat{F}}_{ll}+\gamma^2\sum_{j\ne l} \mathbf{\hat{F}}_{jl}^{\dagger}\mathbf{\hat{F}}_{jl}+\left(\delta_{ll}+\gamma^2\sum_{j\ne l}\delta_{jl}+K\right)\mathbf{I}_M\right)^{-1}\mathbf{\hat{F}}_{ll}^{\dagger},$$
where $\alpha^{opt}$ is such that $\|\mathbf{A}_{l}^{opt}\|^2=1.$ This completes the proof.
\end{proof}

The precoding described above is primarily suited for maximizing the minimum of the rates achieved by all the users. When the performance metric of interest is sum rate, this precoding can be combined with power control, scheduling, and other similar techniques. Since our main concern is the inter-cell interference resulting from pilot contamination, and to avoid too complicated systems, we do not use that possibility in this paper. In the next section, all numerical results and comparisons are performed without power control.

\section{Numerical Results}
\label{sec:numresults}

\begin{figure}[!t]
\centering
\includegraphics[width=150mm]{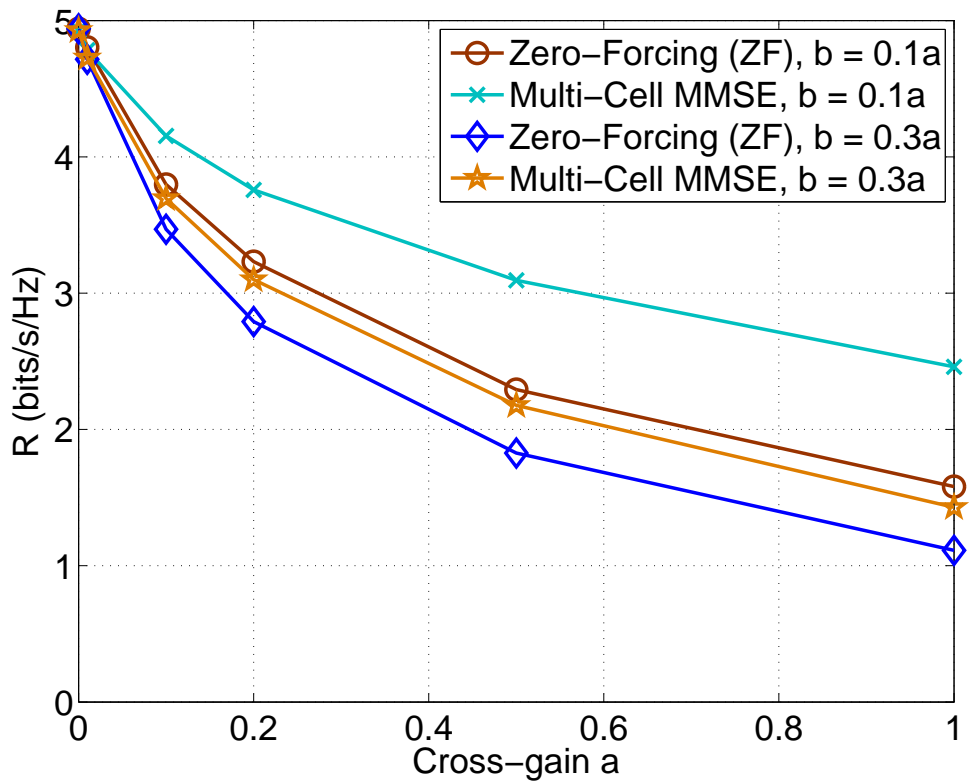}
\caption{Comparison of zero-forcing and multi-cell MMSE precoding methods; $a$ and $b$ correspond to different cross-gains and $R$ denotes the minimum rate achieved by all users}
\label{fig:plot3}
\end{figure}

\begin{figure}[!t]
\centering
\includegraphics[width=150mm]{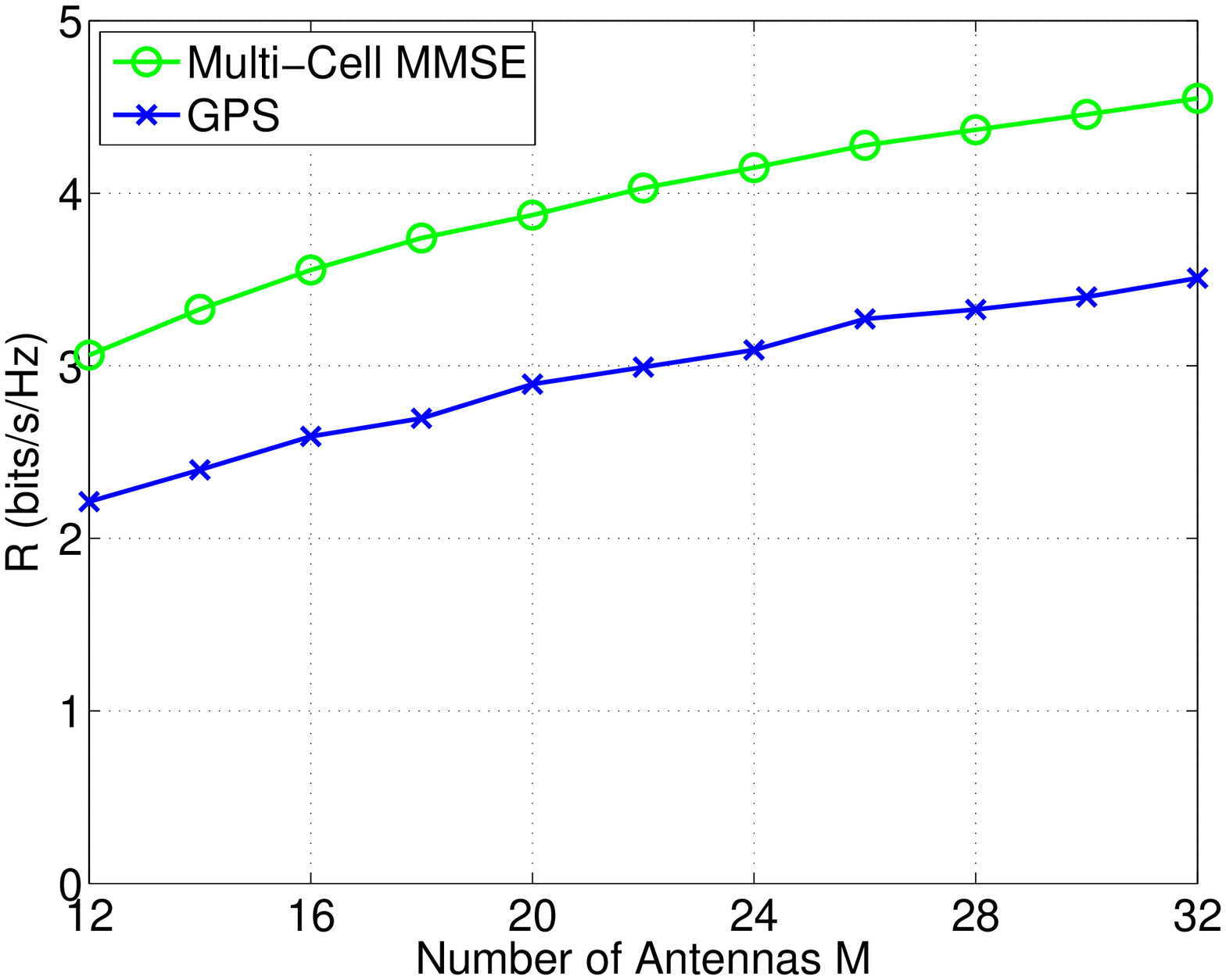}
\caption{Comparison of GPS and multi-cell MMSE precoding methods with $a=0.8$ and $b=0.1a$; $R$ denotes the minimum rate achieved by all users}
\label{fig:L4 min rate}
\end{figure}

We consider a multi-cell system with $L=4$ cells, $M=8$ antennas at all base stations, $K=2$ users in every cell and training length of $\tau=4$. We consider $p_f = 20$ dB and $p_r = 10$ dB. Orthogonal training sequences are collectively used within the $1$-st and $2$-nd cells. The training sequences used in the $1$-st ($2$-nd) cell are reused in the $3$-rd ($4$-th) cell. Thus, we model a scenario where training sequences are reused. We keep the propagation factors as follows: for all $k$, $\beta_{jlk}=1$ if $j=l$, $\beta_{jlk}=a$ if $(j,l) \in \{(1,2),(2,1),(3,4),(4,3)\}$, and $\beta_{jlk}=b$ for all other values of $j$ and $l$. ``Frequency reuse'' is handled semi-quantitatively by adjusting the cross-gains.

The performance metric of interest is the minimum rate achieved by all users denoted by $R = \min_{jk}R_{jk}$. Multi-Cell MMSE precoding denotes the new precoding method developed in this paper given in (\ref{eq:Aoptthm}) with parameter $\gamma$ set to unity. ZF precoding denotes the popular zero-forcing precoding given in (\ref{eq:ZF}). GPS denotes the single-cell precoding method suggested in \cite{GPS08}, which is a special case of the precoding given in (\ref{eq:Aoptthm}) with parameter $\gamma$ set to zero. In Figure \ref{fig:plot3}, we plot the performance of ZF and multi-cell MMSE precoding methods for different values of $a$ and $b$. We observe significant advantage of using multi-cell MMSE precoding for wide range of values of $a$ and $b$. In Figure \ref{fig:L4 min rate}, we plot the performance of GPS and multi-cell MMSE precoding methods as a function of the number of antennas $M$ (same training sequences). In summary, these numerical results show that the new multi-cell MMSE precoding offers significant performance gain over popular single-cell precoding methods.

\section{Conclusion}
\label{sec:conc}
In this paper, we characterize the impact of corrupted channel estimates caused by pilot contamination in TDD systems. When non-orthogonal training sequences are assigned to users, the precoding matrix used at a (multiple antenna) base station becomes correlated with the channel to users in other cells (referred to as pilot contamination). For a special setting that captures pilot contamination, we obtain a closed-form expressions for the achievable rates. Using these analytical expressions, we show that, in the presence of pilot contamination, rates achieved by users saturate with the number of base station antennas. We conclude that appropriate frequency/time reuse techniques have to be employed to overcome this saturation effect. The fact that pilot contamination hasn't surfaced in FDD studies suggests that researchers are assuming partial CSI with independently corrupted noise, and are not fully incorporating the impact of channel estimation.

Next, we develop a multi-cell MMSE-based precoding that depends on the set of training sequences assigned to the users. We obtain this precoding as the solution to an optimization problem whose objective function consists of two parts: (\emph{i}) the mean-square error of signals received at the users in the same cell, and (\emph{ii}) the mean-square interference caused at the users in other cells. We show that this precoding method reduce both intra-cell interference and inter-cell interference, and thus is similar in spirit to existing joint-cell precoding techniques. The primary differences between joint-cell precoding and our approach are that a.) our approach is distributed in nature and, b.) we explicitly take into account the set of training sequences assigned to the users. Through numerical results, we show that our method outperforms popular single-cell precoding methods.


%figure1

%figure2

%figure 5 & 6


\begin{thebibliography}{10}
\providecommand{\url}[1]{#1}
\csname url@samestyle\endcsname
\providecommand{\newblock}{\relax}
\providecommand{\bibinfo}[2]{#2}
\providecommand{\BIBentrySTDinterwordspacing}{\spaceskip=0pt\relax}
\providecommand{\BIBentryALTinterwordstretchfactor}{4}
\providecommand{\BIBentryALTinterwordspacing}{\spaceskip=\fontdimen2\font plus
\BIBentryALTinterwordstretchfactor\fontdimen3\font minus
  \fontdimen4\font\relax}
\providecommand{\BIBforeignlanguage}[2]{{%
\expandafter\ifx\csname l@#1\endcsname\relax
\typeout{** WARNING: IEEEtran.bst: No hyphenation pattern has been}%
\typeout{** loaded for the language `#1'. Using the pattern for}%
\typeout{** the default language instead.}%
\else
\language=\csname l@#1\endcsname
\fi
#2}}
\providecommand{\BIBdecl}{\relax}
\BIBdecl

\bibitem{Jose_AMV_09}
J.~Jose, A.~Ashikhmin, T.~Marzetta, and S.~Vishwanath, ``Pilot contamination
  problem in multi-cell {TDD} systems,'' in \emph{Proc. {IEEE} International
  Symposium on Information Theory (ISIT)}, Jul. 2009, pp. 2184 --2188.

\bibitem{MIMO_wireless_text}
E.~Biglieri, R.~Calderbank, A.~Constantinides, A.~Goldsmith, A.~Paulraj, and
  H.~V. Poor, \emph{{MIMO} Wireless Communications}.\hskip 1em plus 0.5em minus
  0.4em\relax Cambridge University Press, 2010.

\bibitem{SH05}
M.~Sharif and B.~Hassibi, ``On the capacity of {MIMO} broadcast channels with
  partial side information,'' \emph{{IEEE} Trans. Inf. Theory}, vol.~51, pp.
  506--522, Feb. 2005.

\bibitem{DLZ05}
P.~Ding, D.~J. Love, and M.~D. Zoltowski, ``On the sum rate of channel subspace
  feedback for multi-antenna broadcast channels,'' in \emph{Proc. {IEEE}
  Globecom}, Nov. 2005, pp. 2699--2703.

\bibitem{N06}
N.~Jindal, ``{MIMO} broadcast channels with finite-rate feedback,''
  \emph{{IEEE} Trans. Inf. Theory}, vol.~52, pp. 5045--5060, Nov. 2006.

\bibitem{YJG07}
T.~Yoo, N.~Jindal, and A.~Goldsmith, ``Multi-antenna broadcast channels with
  limited feedback and user selection,'' \emph{{IEEE} J. Sel. Areas Commun.},
  vol.~25, pp. 1478--1491, 2007.

\bibitem{Ashikhmin_ISIT07}
A.~Ashikhmin and R.~Gopalan, ``Grassmannian packings for efficient quantization
  in {MIMO} broadcast systems,'' in \emph{Proc. {IEEE} International Symposium
  on Information Theory (ISIT)}, Jun. 2007, pp. 1811 --1815.

\bibitem{HHA07:SDMA_feedback_rate}
K.~Huang, {R. W. Heath, Jr.}, and J.~G. Andrews, ``Space division multiple
  access with a sum feedback rate constraint,'' \emph{{IEEE} Trans. Signal
  Process.}, vol.~55, no.~7, pp. 3879--3891, 2007.

\bibitem{Marzetta:howmuchtraining}
T.~L. Marzetta, ``How much training is required for multiuser {MIMO}?'' in
  \emph{Proc. Asilomar Conference on Signals, Systems and Computers ({ACSSC})},
  Pacific Grove, CA, USA, Oct./Nov. 2006, pp. 359--363.

\bibitem{JAWV08}
J.~Jose, A.~Ashikhmin, P.~Whiting, and S.~Vishwanath, ``Scheduling and
  pre-conditioning in multi-user {MIMO} {TDD} systems,'' in \emph{Proc. {IEEE}
  International Conference on Communications ({ICC})}, Beijing, China, May
  2008, pp. 4100 -- 4105.

\bibitem{MH06}
T.~L. Marzetta and B.~M. Hochwald, ``{Fast transfer of channel state
  information in wireless systems},'' \emph{{IEEE} Trans. Signal Process.},
  vol.~54, pp. 1268--1278, 2006.

\bibitem{Foschini:coordinatingnetwork}
G.~J. Foschini, K.~Karakayali, and R.~A. Valenzuela, ``Coordinating multiple
  antenna cellular networks to achieve enormous spectral efficiency,''
  \emph{IEE Proceedings Communications}, vol. 153, pp. 548--555, Aug. 2006.

\bibitem{VLV07}
S.~Venkatesan, A.~Lozano, and R.~Valenzuela, ``Network {MIMO}: Overcoming
  intercell interference in indoor wireless systems,'' in \emph{Proc. Asilomar
  Conference on Signals, Systems and Computers ({ACSSC})}, Pacific Grove, CA,
  Nov. 2007, pp. 83--87.

\bibitem{SSZ04}
S.~{Shamai (Shitz)}, O.~Somekh, and B.~M. Zaidel, ``Multi-cell communications:
  An information thoeretic perspective,'' in \emph{Joint Workshop on
  Communicaitons and Coding (JWCC)}, Florence, Italy, Oct. 2004.

\bibitem{ZD04}
H.~Zhang and H.~Dai, ``Co-channel interference mitigation and cooperative
  processing in downlink multicell multiuser {MIMO} networks,'' \emph{European
  Journal on Wireless Communications and Networking}, no.~2, pp. 222--235, 4th
  Quarter 2004.

\bibitem{Zhang_CAGH_2009}
J.~Zhang, R.~Chen, J.~G. Andrews, A.~Ghosh, and R.~W. Heath, ``Networked {MIMO}
  with clustered linear precoding,'' \emph{Trans. Wireless. Comm.}, vol.~8,
  no.~4, pp. 1910--1921, 2009.

\bibitem{Caire:gaussianbc}
G.~Caire and S.~{Shamai (Shitz)}, ``On the achievable throughput of a
  multiantenna {G}aussian broadcast channel,'' \emph{{IEEE} Trans. Inf.
  Theory}, vol.~49, pp. 1691--1707, Jul. 2003.

\bibitem{Viswanath:sumcapacity}
P.~Viswanath and D.~N.~C. Tse, ``Sum capacity of the vector {G}aussian
  broadcast channel and uplink-downlink duality,'' \emph{{IEEE} Trans. Inf.
  Theory}, vol.~49, pp. 1912--1921, Aug. 2003.

\bibitem{Vishwanath:duality}
S.~Vishwanath, N.~Jindal, and A.~J. Goldsmith, ``Duality, achievable rates, and
  sum-rate capacity of {G}aussian {MIMO} broadcast channels,'' \emph{{IEEE}
  Trans. Inf. Theory}, vol.~49, pp. 2658--2668, Oct. 2003.

\bibitem{YC04}
W.~Yu and J.~M. Cioffi, ``Sum capacity of {Gaussian} vector broadcast
  channels,'' \emph{{IEEE} Trans. Inf. Theory}, vol.~50, no.~9, pp. 1875--1892,
  2004.

\bibitem{Weingarten:capacity}
H.~Weingarten, Y.~Steinberg, and S.~{Shamai (Shitz)}, ``The capacity region of
  the {G}aussian multiple-input multiple-output broadcast channel,''
  \emph{{IEEE} Trans. Inf. Theory}, vol.~52, pp. 3936--3964, Sep. 2006.

\bibitem{Boccardi:precoding}
F.~Boccardi, F.~Tosato, and G.~Caire, ``Precoding schemes for the {MIMO-GBC},''
  in \emph{Int. Zurich Seminar on Communications}, Feb. 2006.

\bibitem{Airy:transmit_precoding}
M.~Airy, S.~Bhadra, {R. W. Heath, Jr.}, and S.~Shakkottai, ``Transmit precoding
  for the multiple antenna broadcast channel,'' in \emph{Vehicular Technology
  Conference}, vol.~3, 2006, pp. 1396--1400.

\bibitem{HPS05_part2}
B.~M. Hochwald, C.~B. Peel, and A.~L. Swindlehurst, ``A vector-perturbation
  technique for near-capacity multiantenna multiuser communication part {II}:
  Perturbation,'' \emph{{IEEE} Trans. Commun.}, vol.~53, pp. 537--544, Jan.
  2005.

\bibitem{SVH06}
M.~Stojnic, H.~Vikalo, and B.~Hassibi, ``Rate maximization in multi-antenna
  broadcast channels with linear preprocessing,'' \emph{{IEEE} Trans. Wireless
  Commun.}, vol.~5, pp. 2338--2342, Sep. 2006.

\bibitem{Shen:low_complexity}
Z.~Shen, R.~Chen, J.~G. Andrews, {R. W. Heath, Jr.}, and B.~L. Evans, ``Low
  complexity user selection algorithms for multiuser {MIMO} systems with block
  diagonalization,'' \emph{{IEEE} Trans. Signal Process.}, vol.~54, pp.
  3658--3663, Sep. 2006.

\bibitem{JCCU08}
M.~Joham, P.~M. Castro, L.~Castedo, and W.~Utschick, ``{MMSE} optimal feedback
  of correlated {CSI} for multi-user precoding,'' in \emph{Proc. {IEEE}
  International Conference on Acoustics, Speech and Signal Processing ({ICASSP
  '08})}, Las Vegas, NV, Mar. 2008, pp. 3129--3132.

\bibitem{GPS08}
K.~S. Gomadam, H.~C. Papadopoulos, and C.-E.~W. Sundberg, ``Techniques for
  multi-user {MIMO} with two-way training,'' in \emph{Proc. {IEEE}
  International Conference on Communications ({ICC})}, Beijing, China, May
  2008, pp. 3360--3366.

\bibitem{KSH00}
T.~Kailath, A.~H. Sayed, and B.~Hassibi, \emph{Linear Estimation}.\hskip 1em
  plus 0.5em minus 0.4em\relax Prentice Hall, 2000.

\bibitem{Hassibi:howmuchtraining}
B.~Hassibi and B.~M. Hochwald, ``How much training is needed in
  multiple-antenna wireless links?'' \emph{{IEEE} Trans. Inf. Theory}, vol.~49,
  pp. 951--963, Apr. 2003.

\bibitem{Marzetta09}
T.~L. Marzetta, ``The ultimate performance of noncooperative multiuser
  {MIMO},'' submitted to IEEE Trans. on Wireless Communications, Jul. 2009.

\end{thebibliography}
\end{document}